\documentclass[11pt]{amsart}
\usepackage{amscd,amssymb,verbatim,mathrsfs,color,bbold,amsmath,comment}
\newtheorem{theorem}{Theorem}
\newtheorem{corollary}[theorem]{Corollary}
\newtheorem{lemma}[theorem]{Lemma}
\newtheorem{proposition}[theorem]{Proposition}

\theoremstyle{definition}
\newtheorem{example}[theorem]{Example}
\newtheorem{remark}[theorem]{Remark}
\newtheorem{question}[theorem]{Question}

\numberwithin{equation}{section}

\newcommand{\ep}{\varepsilon}

\newcommand{\ES}{\mathrm{ES}}
\newcommand{\Var}{\mathrm{Var}}

\newcommand{\abs}[1]{\lvert#1\rvert}
\newcommand{\norm}[1]{\lVert#1\rVert}
\newcommand{\bigabs}[1]{\bigl\lvert#1\bigr\rvert}
\newcommand{\bignorm}[1]{\bigl\lVert#1\bigr\rVert}
\newcommand{\Bigabs}[1]{\Bigl\lvert#1\Bigr\rvert}
\newcommand{\Bignorm}[1]{\Bigl\lVert#1\Bigr\rVert}

\newcommand{\N}{{\mathbb N}}
\newcommand{\R}{{\mathbb R}}

\newcommand{\C}{{\mathbb C}}
\newcommand{\E}{{\mathbb E}}

\newcommand{\cD}{{\mathcal D}}

\newcommand{\cX}{{\mathcal X}}
\newcommand{\cA}{{\mathcal A}}

\newcommand{\cK}{{\mathcal K}}
\newcommand{\cC}{{\mathcal C}}
\newcommand{\cF}{{\mathcal F}}

\newcommand{\bP}{{\mathbb P}}

\def\one{\mathbb 1}

\title{The strong Fatou property of risk measures}

\author[S.~Chen]{Shengzhong Chen}
\address{Department of Mathematics, Ryerson University, Canada}
\email{sheng.zhong.chen@ryerson.ca}

\author[N.~Gao]{Niushan Gao}
\address{Department of Mathematics, Ryerson University, Canada}
\email{niushan@ryerson.ca}

\author[F.~Xanthos]{Foivos Xanthos}
\address{Department of Mathematics, Ryerson University, Canada}
\email{foivos@ryerson.ca}



\thanks{The authors thank NSERC for financial support.}

\date{\today}

\begin{document}
\maketitle

\begin{abstract}
In this paper, we  explore several Fatou-type properties of risk measures. The paper continues to reveal that the strong Fatou property, which was introduced in \cite{GX:18}, seems to be most suitable to ensure nice dual representations of risk measures. Our main result asserts that every quasiconvex law-invariant functional on a rearrangement invariant space $\cX$ with the strong Fatou property is $\sigma(\cX,L^\infty)$ lower semicontinuous and that the converse is true on a wide range of rearrangement invariant spaces.
We also study inf-convolutions of law-invariant or surplus-invariant risk measures that preserve the (strong) Fatou property.
\end{abstract}

\maketitle

\section{Introduction}

In the early stage of the axiomatic theory of risk measures, the model space $\cX$ is usually taken to be an $L^p$-space. The increasing use of heavily-tailed distributions in risk modelling has led to more general choices of $\cX$, such as Orlicz spaces, Orlicz hearts and other rearrangement invariant spaces (see e.g. \cite{BC:17,BF:09,FG:02,GLMX:18,GLX:16,GM:18,GX:18,LS:18,LS:17,P:13}).
On these model spaces, when one deals with optimization problems, convex duality techniques are desirable and are available as soon as the risk measures involved admit tractable dual representations.
When $\cX=L^p$, this is ensured if the risk measures have the Fatou property.
When $\cX$ is a general Orlicz space $L^\Phi$,  the Fatou property, however,  no longer guarantees tractable dual representations (\cite{GLX:16}). In order to overcome this obstacle, the last two authors of the present paper introduced the strong Fatou property in \cite{GX:18}, which turns out to be the right continuity adjustment in the Orlicz space framework.
This paper continues to investigate Fatou-type properties of risk measures and to highlight the importance of the strong Fatou property.

Through out this paper the model spaces $\cX$ we work on  are \emph{function spaces} over a fixed probability space $(\Omega,\cF,\bP)$, i.e., order ideals of $L^0:=L^0(\Omega,\cF,\bP)$.  Orlicz spaces, including $L^p$ ($1\leq p\leq \infty$), are typical function spaces. We refer to the appendix for some notation and facts on function spaces, in particular, on \emph{rearrangement invariant (r.i.) spaces}. As usual, we do not distinguish two random variables that are almost surely equal.

All functionals $\rho:\cX \rightarrow (-\infty,\infty]$ considered in this paper are \emph{proper}, i.e., not identically $\infty$, unless otherwise stated. $\rho:\cX \rightarrow (-\infty,\infty]$ is {\em convex} if $\rho(\lambda X+(1-\lambda)Y)\leq\lambda\rho(X)+(1-\lambda)\rho(Y)$ for any $X,Y\in \cX$ and $\lambda\in[0,1]$, and is {\em quasiconvex} if the sublevel set $\{\rho\leq m\}:=\{X\in \cX : \rho(X)\leq m\}$ is convex for every $m\in\R$. For a fixed nonzero positive vector $S\in \cX$, $\rho$ is {\em $S$-additive} if $\rho(X+mS)=\rho(X)-m$ for any $X\in\cX$ and $m\in\R$. In the case of $S=\one_\Omega$ we say that $\rho$ is \emph{cash-additive}.
$\rho$ is \emph{law-invariant} if $\rho(X)=\rho(Y)$ whenever $X,Y\in \cX$ have the same distribution, is \emph{surplus-invariant} if $\rho(X)=\rho(-X^-) $ for every $X\in\cX$, and is \emph{surplus-invariant subject to positivity} if $\rho(X)=\rho(-X^-)$ for every $X\in\cX$ such that $\rho(X)>0$.

For a locally convex topology $\tau$ on $\cX$, $\rho:\cX \rightarrow (-\infty,\infty]$ is \emph{$\tau$ lower semicontinuous} if
$\{\rho\leq\lambda\} $ is $\tau$-closed for every $\lambda\in\R$.
Clearly, the coarser $\tau$ is, the stronger the $\tau$ lower semicontinuity is. The well-known Fechel-Moreau duality asserts that a convex functional $\rho:\cX \rightarrow (-\infty,\infty]$ is $\tau$ lower semicontinuous if and only if it admits a dual representation via the topological dual $(X,\tau)^*$. We say that $\rho:\cX \rightarrow (-\infty,\infty]$ has the
\begin{enumerate}
\item \emph{Fatou property} if $\rho(X)\leq \liminf_n\rho(X_n)$ whenever $(X_n)\subset \cX$ and $X\in \cX$ satisfy $X_n\stackrel{o}{\longrightarrow}X$ in $\cX$, i.e., $X_n\xrightarrow{a.s.}X$ and $\abs{X_n}\leq X_0$ for some $X_0\in\cX$ and all $n\in\N$,
\item \emph{super Fatou property} if $\rho(X)\leq \liminf_n\rho(X_n)$ whenever $(X_n)\subset \cX$ and $X\in \cX$ satisfy $X_n\xrightarrow{a.s.}X$,
\item  \emph{strong Fatou property} if $\cX$ carries a norm and $\rho(X)\leq \liminf_n\rho(X_n)$ whenever $(X_n)\subset \cX$ and $X\in \cX$ satisfy $X_n\xrightarrow{a.s.}X$ and $(X_n)$ is norm bounded.
\end{enumerate}
Clearly, the strong Fatou property is intermediate among these three Fatou-type properties, stronger than the Fatou property and weaker than the super Fatou property. It is also clear that the strong Fatou property and the Fatou property coincide on $L^\infty$. Moreover, as is well-known, the Fatou property is generally stronger than norm lower semicontinuity  but coincides with it when the underlying model space $\cX$ has order continuous norm.

It has been well known since \cite{D:02} that a quasiconvex functional $\rho$ on $L^\infty$ is $\sigma(L^\infty,L^1)$ lower semicontinuous if and only if it has the  Fatou property. When $\rho$ is additionally law-invariant, it was proved that $\rho$ has the Fatou property if and only if it is norm lower semicontinuous (\cite{JS:}), if and only if it is $\sigma(L^\infty,L^\infty)$ lower semicontinuous (\cite{FS:12}).
Recently, it was proved in \cite{GLX:16} that a convex functional on an Orlicz space $L^\Phi$ with the Fatou property may fail the $\sigma(L^\Phi,L^\Psi)$ lower semicontinuity, where $\Psi$ is the conjugate function of $\Phi$. Nonetheless, \cite{GX:18} showed that a quasiconvex functional $\rho:L^\Phi\rightarrow(-\infty,\infty]$ has the strong Fatou property if and only if it is $\sigma(L^\Phi,H^\Psi)$ lower semicontinuous, where $H^\Psi$ is the heart of $L^\Psi$.
When $\rho$ is additionally law-invariant,   \cite{GLMX:18} showed that the strong Fatou property of $\rho$ is equivalent to the Fatou property and to $\sigma(L^\Phi,L^\Psi)$ (respectively, $\sigma(L^\Phi,H^\Psi)$, $\sigma(L^\Phi,L^\infty)$) lower semicontinuity, but in general, not to norm lower semicontinuity. Furthermore, if a quasiconvex functional $\rho:\cX\rightarrow(-\infty,\infty]$ is surplus-invariant  or is surplus-invariant subject to positivity and $S$-additive for some $0<S\in\cX$, it is shown in \cite{GM:18} that the strong Fatou property of $\rho$ is equivalent to the Fatou property and to the super Fatou property, and in the case of $\cX=L^\Phi$, they are all equivalent to $\sigma(L^\Phi,L^\Psi)$ (respectively, $\sigma(L^\Phi,H^\Psi)$, $\sigma(L^\Phi,L^\infty)$) lower semicontinuity as well.

The main result of this paper  asserts that any quasiconvex, law-invariant functional $\rho$ on an r.i.\ space $\cX$ with the strong Fatou property is $\sigma(\cX,L^\infty)$ lower semicontinuous (Theorem~\ref{s-Fatou}). We also study the relations between the strong Fatou property, $\sigma(\cX,L^\infty)$ lower semicontinuity, and the Fatou property. We show that the strong Fatou property of a quasiconvex law-invariant functional $\rho$ is ``almost'' equivalent to $\sigma(\cX,L^\infty)$ lower semicontinuity (Proposition~\ref{reverse1}) and that if $\cX$ has order continuous norm and is not equal to $L^1$ then the strong Fatou property of a quasiconvex law-invariant functional $\rho$ is equivalent to both $\sigma(\cX,L^\infty)$ lower semicontinuity and the Fatou property (Proposition~\ref{oc-prop}).
In Section 3, we study the Fatou-type properties of inf-convolutions.  In general the (strong) Fatou property is not preserved by inf-convolution (see, e.g., \cite{D:05}). In \cite{FS:08}, it was proved that the Fatou property is preserved by inf-convolutions of convex, cash-additive, law-invariant functionals on $L^p$. In Proposition ~\ref{law-thm} we extend this result for the strong Fatou property on r.i.\ spaces. In Proposition ~\ref{si-thm}, we derive a similar result for inf-convolutions of convex functionals that are $S$-additive and surplus-invariant subject to positivity.

\section{Law-invariant Functionals}

\emph{Throughout this section we will assume that $\cX$  is an r.i.\ space over a fixed nonatomic probability space $(\Omega,\cF,\bP)$.}
We refer to the appendix for notation and facts on function spaces. Write $\pi$ to denote a
finite measurable partition of $\Omega$ whose members have non-zero
probabilities, and write $\Pi$ for the collection of all such $\pi$.
Denote by $\sigma(\pi)$ the finite $\sigma$-subalgebra
generated by $\pi$, and write $
\mathbb{E}[X|\pi] := \mathbb{E}[X|\sigma(\pi)]$.
For all $X\in\cX$ and $\pi\in\Pi$, we have $\E[X|\pi]\in L^\infty\subset \cX$ by \eqref{contain}, and moreover, by \cite[Theorem~4.8, p.61]{BS:88},
\begin{align}\label{con-e-c}\bignorm{\E[X|\pi]}\leq \norm{X}.\end{align}

%over a non-atomic probability space $(\Omega,\cF,\bP)$. We collect some basic facts about r.i.\ spaces in  Appendix~\ref{R.I.}.
%First of all we mention the following $$L^\infty\subset \cX\subset L^1,\qquad L^\infty\subset \cX_n^\sim \subset L^1.$$
% and that there exists $C>0$ such that
%$$\norm{X}\leq C\norm{X}_\infty$$
%for every $X\in L^\infty$.

Our main result asserts that the strong Fatou property of a quasiconvex law-invariant risk measure implies $\sigma(\cX,L^\infty)$ lower semicontinuity. For this purpose, we need to establish some preliminary technical results.
First of all, recall the
following useful result, which is contained in Step 2 in the proof of
\cite[Lemma~1.3]{S:10}.

\begin{lemma}[\cite{S:10}]\label{bounded}
Let $X\in L^\infty$, $\varepsilon>0$ and $\pi\in\Pi$. Then, there exist
$X_1,\dots,X_N\in L^\infty$ which have the same distribution as $X$ and satisfy
\[
\Bignorm{\frac{1}{N}\sum_{i=1}^NX_i-\mathbb{E}[X|\pi]}_\infty \leq \varepsilon.
\]
\end{lemma}

A sequence $(X_n)\subset \cX$ is said to {\em order converge} to $X\in\cX$, written as $X_n\stackrel{o}{\longrightarrow}X$, if $X_n\xrightarrow{a.s.}X$ and
there exists $X_0\in \cX$ such that $\abs{X_n}\leq X_0$ for all $n\in \N$. We say that a subset $\cC\subset \cX$ is \emph{order closed} in $\cX$ if it contains all the order limits of sequences with terms in it. Clearly, a functional $\rho:\cX\rightarrow(-\infty,\infty]$ has the Fatou property if and only if the sublevel set $\{\rho\leq m\}$ is order closed for every $m\in\R$, and has the strong Fatou property if and only if each sublevel set $\{\rho\leq m\}$ contains the a.s.-limits of norm bounded sequences with terms in it. A set $\cC$ is \emph{law-invariant} if $Y\in \cC$ whenever $X\in\cC$ and $X,Y$ have the same distribution. It is also clear that a functional $\rho$ is law-invariant if and only if each sublevel set $\{\rho\leq  m\}$ is law-invariant.

\begin{proposition}\label{con-con-e}
Let $\cC$ be a convex, order closed, law-invariant set in $\cX$. Then,
$\mathbb{E}[X|\pi]\in \cC$ for any $X\in \cC$ and any $\pi\in\Pi$.
\end{proposition}

\begin{proof}
Let $X\in \cC$, $\pi=\{B_1,\dots,B_k\}\in\Pi$ and fix $n\in\N$. Set $A_n=\{\abs{X}\leq n\}$.
Consider the nonatomic probability space $(A_n,\cF_{|A_n},\bP_{|A_n})$, where $\cF_{|A_n}:=\{B\in\cF : B\subset A_n\}$ and
$\bP_{|A_n}:\cF_{|A_n}\to[0,1]$ is defined by $\bP_{|A_n}(B):=\bP(B|A_n)$. Applying Lemma~\ref{bounded} to $X_{|A_n}$ and the partition $\{B_1\cap A_n,\dots,B_k\cap A_n\}$ of the state space $A_n$, we obtain $X_{n,1}',\dots,X_{n,N_n}'\in L^\infty(A_n,\cF_{|A_n},\bP_{|A_n})$ such that $X'_{n,j}$ has the same distribution as $X_{|A_n}$ for all $1\leq j\leq N_n$ and
\[
\Bigabs{\sum_{i=1}^k\E_{|A_n}\big[X_{|A_n}|B_i\cap A_n\big]\one_{B_i\cap
A_n}-\frac{1}{N_n}\sum_{j=1}^{N_n}X_{n,j}'}\leq \frac{1}{n}\quad\quad \mbox{$\bP_{|A_n}$-a.s.~on $A_n$},
\]
where $\E_{|A_n}$ denotes the expectation under $\bP_{|A_n}$. A direct computation shows that $\E_{|A_n}\big[X_{|A_n}|B_i\cap A_n\big]=\E[X|B_i\cap A_n]$ for all $1\leq i\leq k$. Set $X_{n,j}'=0$ on $A_n^c$. Then
\begin{align}\label{int1}
\Bigabs{\sum_{i=1}^k\E[X|B_i\cap A_n]\one_{B_i\cap A_n}-\frac{1}{N_n}\sum_{j=1}^{N_n}X_{n,j}'}\leq \frac{1}{n}\one_{A_n}\leq \frac{1}{n}\one\stackrel{o}{\longrightarrow}0\;\mbox{ in }\cX.
\end{align}
Set $\delta=\frac{1}{2}\min_{1\leq i\leq k}\bP(B_i)$. Since $A_n\uparrow\Omega$, there exists $n_0\in\N$ such that $\bP(B_i\cap A_n)\geq \delta$ for all $1\leq i\leq k$ and $n\geq n_0$.
Thus, for all $1\leq i\leq k$ and $n\geq n_0$,
\begin{align*}
&\,\Bigabs{\E[X|B_i\cap A_n]\one_{B_i\cap A_n}-\E[X|B_i]\one_{B_i}}\\
\leq&\,\Bigabs{(\E[X|B_i\cap A_n]-\E[X|B_i])\one_{B_i\cap A_n}} + \Bigabs{\E[X|B_i](\one_{B_i}-\one_{B_i\cap A_n})}\\
\leq&\,
\frac{\bigabs{\E[X\one_{B_i\cap A_n}]\bP(B_i)-\E[X\one_{B_i}]\bP(B_i\cap A_n)}}{\bP(B_i\cap A_n)\bP(B_i))}\,\one+ \frac{\E[\abs{X}]}{2\delta}\one_{B_i\cap A_n^c}\\
\leq&\,
\frac{\bigabs{(\E[X\one_{B_i\cap A_n}]-\E[X\one_{B_i}])\bP(B_i)+\E[X\one_{B_i}](\bP(B_i)-\bP(B_i\cap A_n))}}{2\delta^2}\,\one+ \frac{\E[\abs{X}]}{2\delta}\one_{B_i\cap A_n^c}\\
\leq&\,
\frac{\E[\abs{X}\one_{B_i\cap A_n^c}]+\E[\abs{X}]\bP(B_i\cap A_n^c)}{2\delta^2}\,\one+ \frac{\E[\abs{X}]}{2\delta}\one_{B_i\cap A_n^c}.
\end{align*}
Therefore, for $n\geq n_0$,
\begin{align*}
\Bigabs{\sum_{i=1}^k\E[X|B_i\cap A_n]\one_{B_i\cap A_n}-\E[X|\pi]}\leq
\frac{\E[\abs{X}\one_{A_n^c}]+\E[\abs{X}]\bP(A_n^c)}{2\delta^2}\,\one+ \frac{\E[\abs{X}]}{2\delta}\one_{A_n^c}.
\end{align*}
Since $\E[\abs{X}\one_{A_n^c}]\longrightarrow0$, $\bP(A_n^c)\longrightarrow0$ and $\one_{A_n^c}\stackrel{o}{\longrightarrow}0$, we have
\begin{align}\label{int2}
\Bigabs{\sum_{i=1}^k\E[X|B_i\cap A_n]\one_{B_i\cap A_n}-\E[X|\pi]}\stackrel{o}{\longrightarrow}0\;\mbox{ in }\cX.
\end{align}
Set $X_{n,j}=X_{n,j}'+X\one_{A_n^c}$ for $1\leq j\leq N_n$. Then, $X_{n,j}$ has the same distribution as
$X$ and, hence, $X_{n,j}\in \cC$ by law-invariance. Thus $\frac{1}{N_n}\sum_{j=1}^{N_n}X_{n,j}\in \cC$ by convexity of $\cC$. Note that
\begin{align}\label{int3}
\Bigabs{\frac{1}{N_n}\sum_{j=1}^{N_n}X_{n,j}-\frac{1}{N_n}\sum_{j=1}^{N_n}X_{n,j}'}=\abs{X}\one_{A_n^c}\stackrel{o}{\longrightarrow}0\;\mbox{ in }\cX.
\end{align}
Combining \eqref{int1}-\eqref{int3}, we have
$$\Bigabs{\frac{1}{N_n}\sum_{j=1}^{N_n}X_{n,j}-\E[X|\pi]}\stackrel{o}{\longrightarrow}0\;\mbox{ in }\cX.$$
This concludes the proof because $\cC$ is order closed.
\end{proof}

The next preliminary results deal with convergence of conditional expectations. They are both well-known to experts. For the convenience of the reader, we provide a proof of Proposition \ref{norm-exp}.

\begin{lemma}\label{inf-exp}
For any $X\in L^\infty$ and $\ep>0$, there exists $\pi\in\Pi$ such that $\norm{\mathbb{E}[X|\pi]-X}_\infty<\ep$.
\end{lemma}

\begin{proposition}\label{norm-exp}Let $X\in\cX$. The following hold.
\begin{enumerate}
\item\label{norm-exp1}  There exists a sequence $(\pi_n)\subset \Pi$ such that
$\E[X|\pi_n]\stackrel{a.s.}{\longrightarrow}X$.
\item\label{norm-exp2} Suppose that $\cX$ has order continuous norm. There exists a sequence $(\pi_n)\subset\Pi$ such that
$\norm{\E[X|\pi_n]-X}\rightarrow0$ and $\E[X|\pi_n]\stackrel{a.s.}{\longrightarrow}X$.
\end{enumerate}
\end{proposition}

\begin{proof}
Assume first that $\cX$ has order continuous norm. Since $X\one_{\{\abs{X}>n\}} \stackrel{o}{\longrightarrow}0$, it follows that $\norm{X\one_{\{\abs{X}>n\}}}\longrightarrow0$.
Thus, for any $n\in\N$, there exists $m_n\in\N$ such that $$\bignorm{X\one_{\{\abs{X}>m_n\}}}\leq \frac{1}{n}.$$
Since $L^\infty$ continuously embeds  into $\cX$ (see \eqref{iso-contain}), by applying Lemma~\ref{inf-exp}, we get $\pi_n\in\Pi$ such that $$\bignorm{\E[X\one_{\{\abs{X}\leq m_n\}}|\pi_n]-X\one_{\{\abs{X}\leq m_n\}}}\leq\frac{1}{n}.$$
Note also that by \eqref{con-e-c}$$\bignorm{\E\big[X\one_{\{\abs{X}>m_n\}}|\pi_n\big]}\leq \norm{X\one_{\{\abs{X}>m_n\}}}\leq \frac{1}{n}.$$
Therefore, it follows that
\begin{align*}
\bignorm{\E[X|\pi_n]-X}=&\,\Bignorm{\E\big[X\one_{\{\abs{X}>m_n\}}|\pi_n\big]+\E[X\one_{\{\abs{X}\leq m_n\}}|\pi_n]-X\one_{\{\abs{X}\leq m_n\}}-X\one_{\{\abs{X}>m_n\}}}\\
\leq&\, \frac{3}{n}\longrightarrow0.
\end{align*}
Since $\cX$ continuously embeds into $L^1$ (see \eqref{iso-contain}), $\norm{\E[X|\pi_n]-X}_1\longrightarrow0$. For a subsequence $(\pi_{n_k})$, we have $\E[X|\pi_{n_k}]\stackrel{a.s.}{\longrightarrow}X$. Replacing $(\pi_n)$ with $(\pi_{n_k})$, this proves \eqref{norm-exp2}. \eqref{norm-exp1} follows by noting again that $\cX\subset L^1$ and applying \eqref{norm-exp2} to $L^1$.
\end{proof}

Propositions \ref{con-con-e} and \ref{norm-exp} imply the following interesting result, which asserts that quasiconvex, law-invariant functionals may be ``localized'' on $L^\infty$.

\begin{corollary}\label{localize}
Let $\rho_1,\rho_2:\mathcal{X} \rightarrow (-\infty,\infty]$ be two quasiconvex, law invariant functionals each of which either has  the strong Fatou property or is $\sigma(\cX,L^\infty)$ lower semicontinuous.
If $\rho_1$ and $\rho_2$ coincide on $L^\infty$, then $\rho_1=\rho_2$. 

\end{corollary}
\begin{proof}
Fix any $X\in \cX$. By Proposition~\ref{norm-exp} applied to $L^1$, we can find  a sequence
$(\pi_n)\subset\Pi$ such that
$\mathbb{E}[X|\pi_n]\longrightarrow X$ \emph{both in $L^1$-norm and almost surely}.
Then, clearly, $\mathbb{E}[X|\pi_n]\xrightarrow{\sigma(\cX,L^\infty)}X$. Thus if $\rho_1$ is $\sigma(\cX,L^\infty)$ lower semicontinuous, we have
\[
\rho_1(X) \leq \liminf_{n\to\infty} \rho_1(\mathbb{E}[X|\pi_n]).
\]
Alternatively, if $\rho_1$ has the strong Fatou property, then $\sup_n\bignorm{\E[X|\pi_n]}\leq \norm{X}$ again implies
\[
\rho_1(X) \leq \liminf_{n\to\infty} \rho_1(\mathbb{E}[X|\pi_n]).
\]
On the other hand, the set $\cC=\big\{Y\in \cX:\rho_1(Y)\leq \rho_1(X)\big\}$ is convex,
law-invariant, and clearly contains $X$. If $\rho_1$ has the strong Fatou property, and therefore, the Fatou property, then $\cC$ is order closed. If $\rho_1$ is $\sigma(\cX,L^\infty)$ lower semicontinuous, then $\cC$ is $\sigma(\cX,L^\infty)$-closed, and is thus also order closed, since order convergence in $\cX$ implies order convergence in $L^1$ (thanks to $\cX\subset L^1$), which in turn implies $\sigma(\cX,L^\infty)$ convergence. Hence, by
Proposition~\ref{con-con-e}, we have $\mathbb{E}[X|\pi_n]\in \cC$ for every
$n\in\N$, so that
\[
\limsup_{n\to\infty}\rho_1(\mathbb{E}[X|\pi_n]) \leq \rho_1(X).
\]
It follows that
\begin{align}
\label{rep-loc}\rho_1(X)=\lim_n\rho_1(\mathbb{E}[X|\pi_n]).
\end{align}
The same conclusion holds for $\rho_2$ as well. Since $\mathbb{E}[X|\pi_n]\in
L^\infty$ for every $n\in\N$ and $\rho_1$ and $\rho_2$ coincide on $L^\infty$, we conclude that
$\rho_1(X)=\rho_2(X)$.
\end{proof}

We are now ready to present our main result.

\begin{theorem}\label{s-Fatou}
Let $\rho:\mathcal{X} \rightarrow (-\infty,\infty]$ a quasiconvex, law-invariant functional that has the strong Fatou property. Then $\rho$ is $\sigma(\cX,L^\infty)$ lower semicontinuous.
If $\rho$ is additionally convex, then it extends uniquely to a convex, law-invariant functional on $L^1$ with the Fatou property. The extension preserves also cash additivity.
\end{theorem}

\begin{proof}
Pick any $m\in\R$ and put $\cC=\{\rho\leq m\}$. Then $\cC$ is order  closed. We show that $\cC$ is $\sigma(\cX,L^\infty)$-closed. Take any net $(X_\alpha)\subset \cC$ and $X\in \cX$ such that $X_\alpha\xrightarrow{\sigma(\cX,L^\infty)}X$. Then $\E[X_\alpha\one_B]\longrightarrow\E[X\one_B]$ for any $B\in\cF$. Consequently, for any $\pi=\{B_1,\dots,B_k\}\in\Pi$,
$$\E[X_\alpha|\pi] = \sum_{i=1}^k\frac{\E\big[X_\alpha\one_{B_i}\big]}{\bP(B_i)}\one_{B_i}\longrightarrow  \sum_{i=1}^k\frac{\E\big[X\one_{B_i}\big]}{\bP(B_i)}\one_{B_i}
=\E[X|\pi],$$
in the $L^\infty$-norm. We can thus take countably many $(\alpha_n)$ such that
$$\bigabs{\E[X_{\alpha_n}|\pi]-\E[X|\pi]}\leq \frac{1}{n}\one\stackrel{o}{\longrightarrow}0\mbox{ in }\cX.$$
Since $\E[X_\alpha|\pi]\in\cC$ for all $\alpha$ by Proposition~\ref{con-con-e}, order closedness of $\cC$ implies that $\mathbb{E}[X|\pi]\in \cC$. It follows that $\rho(\E[X|\pi])\leq m$ for any $\pi\in\Pi$.
Now by Proposition~\ref{norm-exp}, we can take $(\pi_n)\subset \Pi$ such that $\E[X|\pi_n]\xrightarrow{a.s.}X$. Since $\sup_n\bignorm{\E[X|\pi_n]}\leq \norm{X}<\infty$, the strong Fatou property of $\rho$ implies  that $\rho(X)\leq\liminf_n\rho(\E[X|\pi_n])\leq m$, so that $X\in \cC$. This proves that $\cC$ is $\sigma(\cX,L^\infty)$-closed. Since $m\in\mathbb{R}$ is arbitrary, $\rho$ is $\sigma(\cX,L^\infty)$ lower semicontinuous.

Now, assume that $\rho$ is convex.
It is clear that $\rho|_{L^\infty}$ is convex and law invariant and has the strong Fatou property. \eqref{rep-loc} implies that it is also proper.
Thus, by \cite[Theorem~2.2]{FS:12}, $\rho|_{L^\infty}$ admits a \emph{unique} convex, law-invariant
extension $\overline{\rho}:L^1\rightarrow(-\infty,\infty]$ that is norm lower semicontinuous, and thus, is $\sigma(L^1,L^\infty)$ lower semicontinuous and has the Fatou property.  Put $\rho^*=\overline{\rho}|_{\cX}$. Since $\rho$ and $\rho^*$ are both $\sigma(\cX,L^\infty)$ lower semicontinuous and coincide on $L^\infty$, $\rho=\rho^*$ by Corollary~\ref{localize}, so that $\overline{\rho}$ extends $\rho$. If $\rho|_{L^\infty}$ is cash additive, then $\overline{\rho}$ is also cash additive.
\end{proof}

\begin{example}\label{bigexamp}
\begin{enumerate}
\item\label{bigexamp1} Without law-invariance, the strong Fatou property may not imply $\sigma(\cX,L^\infty)$ lower semicontinuity. Consider $\cX=L^2$. Take $Z\in L^2\backslash L^\infty$, and put $\rho(X)=\E[XZ]$ for every $X\in L^2$. Being linear, $\rho$ is $\sigma(L^2,L^\infty)$ lower semicontinuous, if and only if, it is $\sigma(L^2,L^\infty)$ continuous, if and only if, $Z\in L^\infty$ by \cite[Theorem 3.16]{AB:06}. Thus $Z\notin L^\infty$ implies that $\rho$ is not $\sigma(L^2,L^\infty)$ lower semicontinuous. However, $\rho$ has the strong Fatou property. Indeed, let $(X_n)\subset \L^2$ and $X\in L^2 $ be such that $M:=\sup_n\norm{X_n}_2<\infty$ and $X_n\xrightarrow{a.s.}X$. We show that $\rho(X_n)=\E[X_nZ]\longrightarrow\E[XZ]=\rho(X)$.  Replacing $X_n$ with $X_n-X$, we may assume that $X=0$.
Suppose otherwise that $\E[X_nZ]\not\rightarrow0$. By passing to a subsequence, we may assume that $\abs{\E[X_nZ]}\geq \delta$ for some $\delta>0$ and all $n\in\N$. Since $X_n\xrightarrow{a.s.}0$, we can find a subsequence $(X_{n_k})$ such that $\bP(\abs{X_{n_k}}\geq \frac{1}{k})\leq \frac{1}{k}$. Then $Z^2\one_{\{\abs{X_{n_k}}\geq \frac{1}{k}\}}\longrightarrow0$ in probability and is dominated by $Z^2\in L^1$.  Dominated Convergence Theorem implies that $\bignorm{Z\one_{\{\abs{X_{n_k}}\geq \frac{1}{k}\}}}_2=\bignorm{Z^2\one_{\{\abs{X_{n_k}}\geq \frac{1}{k}\}}}_1^{\frac{1}{2}}\longrightarrow0$. It follows that
\begin{align*}
\bigabs{\E[X_{n_k}Z]}
\leq &\Bigabs{\E\big[X_{n_k}Z\one_{\{\abs{X_{n_k}}<\frac{1}{k}\}}\big]}+
\Bigabs{\E\big[X_{n_k}Z\one_{\{\abs{X_{n_k}}\geq \frac{1}{k}\}} \big]}\\
\leq&\bignorm{X_{n_k}\one_{\{\abs{X_{n_k}}< \frac{1}{k}\}}}_2\norm{Z}_2+
\norm{X_{n_k}}_2\bignorm{Z\one_{\{\abs{X_{n_k}}\geq \frac{1}{k}\}}}_2\\
\leq & \frac{1}{k}\norm{Z}_2+M\bignorm{Z\one_{\{\abs{X_{n_k}}\geq \frac{1}{k}\}}}_2
\longrightarrow 0.
\end{align*}
This contradiction completes the proof.
\item\label{bigexamp2} The extended functional $\overline{\rho}$ on $L^1$ may not have the strong Fatou property. Set $\rho(X)=\E[X]$ on $L^\infty$ and $\overline{\rho}(X)=\E[X]$ on $L^1$, respectively. Clearly, $\rho$ has the strong Fatou property, and $\overline{\rho}$ is the unique convex, law-invariant extension of $\rho$ on $L^1$ that has the Fatou property. But $\overline{\rho}$ does not have the strong Fatou property. Indeed, in view of nonatomicity, take a decreasing sequence of measurable sets $(A_n)$ such that $\bP(A_n)=\frac{1}{n}$ for any $n\in\N$. Set $X_n=-n\one_{A_n}$ for every $n\in\N$. Then $\norm{X_n}=1$ for all $n\in \N$, $X_n\xrightarrow{a.s.}0$, but $\liminf_n\E[X_n]=-1<0=\E[0]$.
\end{enumerate}
\end{example}

Clearly, the proof of Theorem~\ref{s-Fatou}  as well as that of Corollary~\ref{localize} heavily relies on Propositions \ref{con-con-e} and \ref{norm-exp}. The following questions are natural directions of possible improvements of these two propositions. A positive answer to the second question on an r.i.\ space $\cX$ would imply that quasiconvex, law-invariant functionals on $\cX$ with the Fatou property are $\sigma(\cX,L^\infty)$ lower semicontinuous. Both of these questions have positive answers in Orlicz spaces; see \cite{GLMX:18}.

\begin{question}\label{question}
\begin{enumerate}
\item Does Proposition~\ref{con-con-e} hold for norm closed sets?
\item Does Proposition~\ref{norm-exp}(2) hold for order convergence of $(\E[X|\pi_n])$ without the assumption that $\cX$ has order continuous norm?
\end{enumerate}
\end{question}

\smallskip

%\begin{corollary}\label{conc-col}
%Let $\rho:\mathcal{X} \rightarrow (-\infty,\infty]$ a proper, convex  functional with the strong Fatou property. Then the following are equivalent
%\begin{itemize}
%\item[(1)] $\rho$ is law invariant,
%\item[(2)] $\succeq_c$-monotone.
%\end{itemize}
%\end{corollary}

%\begin{proof}
%$(2) \Rightarrow (1)$ is clear. For $(2) \Rightarrow (1)$  consider the extension of $\rho$ to $L^1$ and apply (\cite{D:05}, Theorem 4.1)
%\end{proof}

%\begin{question}Let $\rho(X)=\E[-X]$ for $X\in L^\infty$. It does not extend to a proper, convex, law-invariant functional on $L^1$ that has the strong Fatou property.
%Does $L^1$ admits a convex, law-invariant risk measure with the strong Fatou property?
%\end{question}

We now turn to study the relations between the strong Fatou property, the Fatou property, and $\sigma(\cX,L^\infty)$ lower semicontinuity. Recall that order convergence in $\cX$ implies order convergence in $L^1$ and thus $\sigma(\cX,L^\infty)$-convergence. Therefore, for \emph{any} functional $\rho:\cX\rightarrow(-\infty,\infty]$, if it is  $\sigma(\cX,L^\infty)$ lower semicontinuous, then it has the Fatou property. In particular, for any quasiconvex, law invariant functional $\rho:\mathcal{X} \rightarrow (-\infty,\infty]$, the following implications hold:
$$\text{strong Fatou property}\quad\implies\quad \sigma(\cX,L^\infty)\text{ lower semicontinuity}\quad \Longrightarrow\quad \text{Fatou property}.$$
The converse of the first implication, although not universally true (cf.~Example~\ref{bigexamp}\eqref{bigexamp2}), can be established \emph{without} law-invariance of $\rho$, under an additional but mild condition on $\cX$, which essentially excludes only $L^1$ among all classical spaces. We say that $\cX$ has Property $(*)$ if
$$\lim_{\mathbb{P}(A)\rightarrow 0}\norm{\one_A}_*=0.$$
Proposition~\ref{uo-dual} shows that it is satisfied by all Orlicz spaces and all r.i.\ spaces with order continuous norm, that are not equal to $L^1$. In particular, it is satisfied by all Orlicz hearts that are not equal to $L^1$, since they are r.i.\ spaces with  order continuous norm.

\begin{proposition}\label{reverse1}
Suppose that $\cX$ satisfies Property $(*)$.  Let $\rho:\mathcal{X} \rightarrow (-\infty,\infty]$ be $\sigma(\cX,L^\infty)$ lower semicontinuous. Then $\rho$ has the strong Fatou property.
\end{proposition}

\begin{proof}
Suppose $(X_n)$ is a norm bounded sequence in $\cX$ that a.s.-converges to $X\in\cX$. By Proposition~\ref{p-star-purpose}, it follows that $$\bigabs{\E[X_nZ]-\E[XZ]}\leq \norm{Z}_\infty\bigabs{\E[\abs{X_n-X}]}\longrightarrow0,$$for any $Z\in L^\infty$. Thus $X_n\xrightarrow{\sigma(\cX,L^\infty)}X$, and $\sigma(\cX,L^\infty)$ lower semicontinuity of $\rho$ implies that $\rho(X)\leq\liminf_n\rho(X_n)$. Hence, $\rho$ has the strong Fatou property.
\end{proof}

\begin{remark}
Let $\cX$ be an r.i.\ space with Property $(*)$. Then, for any functional $\overline{\rho}:L^1 \rightarrow (-\infty,\infty]$ with the Fatou property, the restriction  of $\overline{\rho}$ on $\cX$ is $\sigma(\cX,L^\infty)$ lower semicontinuous and thus  has the strong Fatou property. This fact in conjunction with Theorem \ref{s-Fatou} reveals that there is a one-to-one correspondence between  convex law-invariant
 risk measures on $\cX$ with the strong Fatou property and convex law-invariant  risk measures on $L^1$ with the Fatou property.
\end{remark}

The converse of the second implication fails without law-invariance (cf.~Example~\ref{bigexamp}\eqref{bigexamp1}). Under law-invariance, the question remains open to us (cf.~Question~\ref{question} and Example~\ref{orlicz-known} below). When $\cX$ has order continuous norm, we show that all the reverse implications hold.

\begin{proposition}\label{oc-prop}
Suppose that $\cX$ has order continuous norm and $\cX\neq L^1$. Let $\rho:\cX\rightarrow (-\infty,\infty]$ be a quasiconvex, law-invariant functional.
The following are equivalent:
\begin{enumerate}
\item\label{oc-prop1} $\rho$ has the strong Fatou property.
\item\label{oc-prop2} $\rho$ is $\sigma(\cX,L^\infty)$ lower semicontinuous.
\item\label{oc-prop3} $\rho$ has the Fatou property.
\end{enumerate}
\end{proposition}

\begin{proof}
By Proposition~\ref{uo-dual}, $\cX$ has Property $(*)$, and thus \eqref{oc-prop1}$\iff$\eqref{oc-prop2}. It suffices to prove \eqref{oc-prop3}$\implies$\eqref{oc-prop2}. Suppose that $\rho$ has the Fatou property.
Pick any $m\in\R$ and put $\cC=\{\rho\leq m\}$. Being order closed, $\cC$ is norm closed (cf.~e.g., \cite[Lemma~3.11]{GX:14}). We show that $\cC$ is $\sigma(\cX,L^\infty)$-closed. Take any net $(X_\alpha)\subset \cC$ and $X\in \cX$ such that $X_\alpha\xrightarrow{\sigma(\cX,L^\infty)}X$.  As in the proof of Theorem~\ref{s-Fatou}, $\mathbb{E}[X|\pi]\in \cC$ for any $\pi\in\Pi$. Thus, it follows from Proposition~\ref{norm-exp}  that $X\in \cC$. This proves that $\cC$ is $\sigma(\cX,L^\infty)$-closed. Since $m\in\mathbb{R}$ is arbitrary, $\rho$ is $\sigma(\cX,L^\infty)$ lower semicontinuous.
\end{proof}

We look at quasiconvex, law-invariant functionals on Orlicz spaces and Orlicz hearts.

\begin{example}\label{orlicz-known}
Let $\rho:\cX\rightarrow (-\infty,\infty]$ be a quasiconvex, law-invariant functional.
\begin{enumerate}
\item Let $\cX$ be an Orlicz space $L^\Phi$ that is not equal to $L^1$. It was shown in \cite[Theorem 2.4]{GX:18} that $\rho$ has the strong Fatou property if and only if it is  $\sigma(L^\Phi,H^\Psi)$ lower semicontinuous. Moreover, \cite[Theorem 1.1]{GLMX:18} shows that $\rho$ is $\sigma(L^\Phi,H^\Psi)$ lower semicontinuous, if and only if, it is $\sigma(L^\Phi,L^\Psi)$ (respectively, $\sigma(L^\Phi,L^\infty)$) lower semicontinuous,  if and only if,  it has the Fatou property. The equivalence of the strong Fatou property and $\sigma(L^\Phi,L^\infty)$ lower semicontinuity also follows from Theorem~\ref{s-Fatou} and Proposition~\ref{reverse1}.
     %However,  note that these properties may not be equivalent to norm-lsc (see \cite[Theorem 1.2]{GLMX:18}.
\item Let $\cX$ be an Orlicz heart $H^\Phi$ that is not equal to $L^1$. By Proposition~\ref{oc-prop}, $\rho$ has the strong Fatou property, if and only if, it is $\sigma(L^\Phi,L^\infty)$ lower semicontinuous, if and only if, it has the Fatou property. Since $H^\Phi$ has order continuous norm and $(H^\Phi)^*=L^\Psi$, the Fatou property is equivalent to norm lower semicontinuity and thus to $\sigma(H^\Phi,L^\Psi)$ lower semicontinuity. Since $L^\infty\subset H^\Psi\subset L^\Psi$, these properties are equivalent to $\sigma(H^\Phi,H^\Psi)$ lower semicontinuity. (Note that, without law-invariance, the strong Fatou property may not imply $\sigma(H^\Phi,H^\Psi)$ lower semicontinuity (see \cite{GLX:16})).
\end{enumerate}
\end{example}

Let's consider the Expected Shortfall.
For $\alpha\in(0,1)$, define Value-at-Risk at level $\alpha$  by
$$\mathrm{Var}_\alpha(X)=\inf\big\{m\in\R: \bP(X+m<0)\leq \alpha\big\},\quad X\in L^0.$$
For $\alpha\in(0,1]$, define Expected Shortfall at level $\alpha$ by
$$\mathrm{ES}_\alpha(X)=\frac{1}{\alpha}\int_0^\alpha \mathrm{Var}_\beta(X)\mathrm{d}\beta,\quad X\in L^1.$$

\begin{example}\label{ES-Fatou}
$\mathrm{ES}_1 $ has the Fatou property but not the strong Fatou property on $L^1$ (cf.~Example~\ref{bigexamp}).
However, when $\alpha\in (0,1)$, the Expected Shortfall does have the strong Fatou property on any r.i.\ space. In fact, it has the super Fatou property on $L^1$.
%, namely, $\mathrm{ES}_\alpha(X)\leq \liminf_n\mathrm{ES}_\alpha(X_n)$ whenever $(X_n)\subset L^1$ and $X\in L^1$ satisfies $X_n\xrightarrow{a.s.}X$.
We include the proof for the sake of completeness. Fix any $\varepsilon\in (0,1-\alpha)$.
By Egorov's Theorem, there exist a measurable set $ B$ and $n_0 \in \mathbb{N}$ such that
$$ \mathbb{P}(B) < \varepsilon,\quad \text{ and } \quad \abs{X_n-X} < \ep \;\;\text{ on } B^c\;\; \text{ for all }  n\geq n_0.$$
Pick any $\beta\in(0,\alpha)$, and take any $n\geq n_0$. Let $m := \Var_{\beta}(X_n) $ and $ m': = m + \ep$. It follows from
$\lbrace X + m'<0 \rbrace  \subseteq   (\lbrace X + m'<0 \rbrace \cap B) \cup (\lbrace X + m'<0 \rbrace \cap \lbrace X_n <X+\varepsilon \rbrace)
 \subseteq   B\cup  \lbrace X_n +m<0\rbrace$
that $
\mathbb{P}(X + m'<0) \leq \mathbb{P}(B)+ \mathbb{P}( X_n +m<0) \leq \varepsilon +\beta$, and consequently,
\begin{align*}
\Var_{\beta+\varepsilon}(X)\leq m'=\Var_{\beta} (X_n) + \ep.
\end{align*}
Since this holds for any $\beta\in(0,\alpha)$ and any $n\geq n_0$, integrating with respect to $\beta$ over $(0,\alpha)$ implies
$\frac{1}{\alpha} \int^{\alpha+\varepsilon}_{\varepsilon} \Var_{\beta}(X)\,\mathrm{d}\beta=\frac{1}{\alpha} \int^\alpha_0 \Var_{\beta+\ep}(X)\,\mathrm{d}\beta \leq\frac{1}{\alpha} \int^{\alpha}_0 \Var_{\beta} (X_n)\,\mathrm{d}\beta+ \ep
=\ES_\alpha(X_n)+\ep$ for all $n\geq n_0$. Taking infimum over $n\geq n_0$, we have
$$\frac{1}{\alpha} \int^{\alpha+\varepsilon}_{\varepsilon} \Var_{\beta}(X)\,\mathrm{d}\beta\leq \inf_{n\geq n_0}\ES_\alpha(X_n)+\ep\leq \liminf_n \ES_\alpha(X_n)+\ep.$$
Now, since $\Var_{\bullet}(X)\in L^1(0,1]$, letting $\ep\rightarrow 0$, we have $\ES_\alpha(X)\leq\liminf_n\ES_\alpha(X_n)$.
\end{example}

\section{Inf-convolutions}

Let $\cX$ be a function space over a probability space. Given the functionals $\rho_i:\cX \rightarrow (-\infty,\infty]$, $ i=1,\dots,d$, their \emph{inf-convolution} is defined by
$$\Box_{i=1}^d \rho_i(X)=\inf\Big\{\sum_{i=1}^d \rho_i(X_i):\;\; X_i \in \cX, i=1,\dots,d, \text{ and } \sum_{i=1}^d X_i=X\Big\},\quad X\in\cX.$$
It is said to be  \emph{exact} if the infimum is attained at every $X\in\cX$.
One can easily check from definition that the inf-convolution $\Box_{i=1}^d \rho_i$ is convex if each $\rho_i$ is convex, and is $S$-additive (respectively, monotone) if some $\rho_i$ is $S$-additive (respectively, monotone). Recall that a functional $\rho:\cX\rightarrow(-\infty,\infty]$ is \emph{monotone} if $\rho(X)\leq \rho(Y)$ whenever $X,Y\in\cX$ satisfy $X\geq Y$.
Using $\rho(X)=\inf\big\{m\in \R:X+mS\in\{\rho\leq 0\}\big\}$, one sees that quasiconvex $S$-additive functionals are convex. Thus we state the results in this section for convex functionals.

The study of inf-convolutions within the framework of risk measure theory was  initiated in \cite{BE:05}. Inf-convolutions of law-invariant functionals have been studied in many papers, see, e.g, \cite{A:09,ELW:17,FS:08,JS:08,LS:18,MG:15} and the references therein. In particular, \cite[Theorem 2.5]{FS:08} asserts that inf-convolutions of convex, cash-additive, law-invariant functionals on $L^p$ ($1\leq p\leq\infty$) that are norm lower semicontinuous, or equivalently, have the Fatou property, are exact and law invariant and have the Fatou property. The following proposition extends this result to r.i.\ spaces.

\begin{proposition}\label{law-thm}
Let $\cX$ be an r.i.\ space over a nonatomic probability space, and $\rho_i:\cX \rightarrow (-\infty,\infty]$, $i=1,\dots,d$, be convex, cash-additive, law-invariant functionals with the strong Fatou property. Then $\Box_{i=1}^d \rho_i:\cX \rightarrow (-\infty,\infty]$ is convex, cash-additive, law-invariant, and exact, and has the strong Fatou property. Moreover, for each $X \in \cX$ there exist increasing functions $f_i:\mathbb{R} \rightarrow \mathbb{R}$, $i=1,\dots,d$, such that $\sum_{i=1}^d f_i(x)=x$ for every $x\in\R$ and
$$\Box_{i=1}^d \rho_i(X)=\sum_{i=1}^d \rho_i(f_i(X)).$$
\end{proposition}

\begin{proof}
By induction, we may assume that $d=2$. By Theorem~\ref{s-Fatou}, each $\rho_i$ extends to a functional $\overline{\rho_i}: L^1 \rightarrow (-\infty,\infty]$ that is convex, cash additive,  law invariant,  and $||\cdot||_1$ lower semicontinuous. Let  $\overline{\rho_1} \Box \overline{\rho_2}:L^1 \rightarrow (-\infty,\infty]$ be the inf-convolution of $\overline{\rho_1}$ and $\overline{\rho_2}$. Clearly,
$$\overline{\rho_1} \Box \overline{\rho_2} (X)\leq \rho_1\Box \rho_2(X)\mbox{ for any }X\in\cX.$$
Now, pick any $X\in\cX$. By \cite[Theorem 2.5]{FS:08}, there exist increasing functions $f_1,f_2:\mathbb{R} \rightarrow \mathbb{R}$ such that $f_1(x)+f_2(x)=x$ for each $x \in \mathbb{R}$ and $\overline{\rho_1} \Box \overline{\rho_2} (X)=\overline{\rho_1}(f_1(X))+\overline{\rho_2}(f_2(X))$. Since $\rho_1,\rho_2$ are cash-additive, without loss of generality we may assume that $f_1(0)=f_2(0)=0$. One easily sees that $f_1$ and $f_2$ are $1$-Lipschitz functions and thus $\abs{f_i(X)} \leq \abs{X}$ for $i=1,2$. Since $\cX$ is an order ideal of $L^0$, we have that $f_i(X)\in \mathcal{X}$ for $i=1,2$. Therefore,
$$\overline{\rho_1} \Box \overline{\rho_2} (X)=\overline{\rho_1}(f_1(X))+\overline{\rho_2}(f_2(X))=\rho_1(f_1(X))+\rho_2(f_2(X)) \geq \rho_1\Box \rho_2(X).$$
It follows that $\overline{\rho_1} \Box \overline{\rho_2} (X)=\rho_1(f_1(X))+\rho_2(f_2(X))=\rho_1\Box \rho_2(X)$, implying that $\rho_1\Box \rho_2$ is exact and $\overline{\rho_1} \Box \overline{\rho_2}$ extends $\rho_1\Box \rho_2$. By \cite[Theorem 2.5]{FS:08}, $\overline{\rho_1} \Box \overline{\rho_2}$, and therefore $\rho_1\Box \rho_2$, is law-invariant.

In remains to show that $\rho_1\Box \rho_2$ has the strong Fatou property. Pick an arbitrary $m\in\R$, and consider the sublevel set $\cC:=\{X \in \cX: \rho_1\Box \rho_2(X) \leq m\} $. Let $(X_n)$ be a norm bounded sequence in $\cC$ that a.s.-converges to $X \in \cX$. It suffices to show that $X\in\cC$. By the exact solution described above, we can find $Y_n, Z_n \in \cX$ with $X_n=Y_n+Z_n$, $\abs{Y_n}\leq \abs{X_n}$, $\abs{Z_n} \leq \abs{X_n}$, and $\rho_1 \Box \rho_2(X_n)=\rho_1(Y_n)+\rho_2(Z_n)$.  Note that $(Y_n),(Z_n)$ are norm bounded sequences in $\cX$. Applying Proposition~\ref{Komlos}\eqref{Komlos2} twice, we can find strictly increasing $(n_j)$ and two random variables $Y,Z\in L^0$ such that $\frac{1}{k}\sum_{j=1}^kY_{n_j}\xrightarrow{a.s.} Y$ and $\frac{1}{k}\sum_{j=1}^kZ_{n_j}\xrightarrow{a.s.} Z$. Since $\abs{\frac{1}{k}\sum_{j=1}^kY_{n_j}} \leq \frac{1}{k}\sum_{j=1}^k\abs{X_{n_j}} \xrightarrow{{a.s.}} \abs{X}$, we get that $\abs{Y}\leq \abs{X} $, so that $Y\in \cX$. Similarly, we have $Z \in \cX$. Note also that $Y+Z=X$ and that $(\frac{1}{k}\sum_{j=1}^kY_{n_j})$ and $(\frac{1}{k}\sum_{j=1}^kZ_{n_j})$ are both norm bounded sequences in $\cX$. Thus, applying the strong Fatou property and convexity of $\rho_i$'s, we get that
\begin{align*}
\rho_1 \Box \rho_2(X) \leq & \rho_1(Y)+\rho_2(Z) \leq \liminf_k \rho_1\Big(\frac{1}{k}\sum_{j=1}^kY_{n_j}\Big)+\liminf_k \rho_2\Big(\frac{1}{k}\sum_{j=1}^kZ_{n_j}\Big)\\
\leq &\liminf_k \frac{\sum_{j=1}^k\rho_1(Y_{n_j})}{k}+\liminf_k \frac{\sum_{j=1}^k\rho_2(Z_{n_j})}{k} \\
\leq &\liminf_k\Big(\frac{\sum_{j=1}^k\rho_1(Y_{n_j})}{k}+ \frac{\sum_{j=1}^k\rho_2(Z_{n_j})}{k} \Big)=\liminf_k\Big(\frac{\sum_{j=1}^k
 \rho_1\Box \rho_2(X_{n_j})}{k}\Big) \leq m.
 \end{align*}
This proves that $X\in\cC$ and completes the proof of the proposition.
\end{proof}

We now turn to study the (super) Fatou property of inf-convolutions of convex $S$-additive functionals that are surplus-invariant subject to positivity. Such functionals are systematically studied in \cite{GM:18}. In particular, \cite[Theorem~29]{GM:18} asserts that the Fatou property and the super Fatou property coincide for such functionals.

Let $\cX$ be a function space over a fixed probability space. A set $\cA\subset \cX$ is \emph{surplus-invariant} if $-X^-\in \cA$ whenever $X\in \cA$, and is  \emph{monotone} if $Y\in\cA$ whenever $Y\geq X$, $Y\in \cX$ and $X\in \cA$. By \cite[Proposition 2]{GM:18}, a set $\cA\subset \cX$ is surplus-invariant and monotone if and only if $\cA=\cX_+-\cD$ for some $\cD\subset \cX_+$ that is \emph{solid in $\cX_+$}, i.e., $Y\in \cD$ whenever $0\leq Y\leq X$ for some $X\in \cD$. Moreover,
%In this case, $\cD=\{-X:X\in \cA,X\leq 0\}$, so that
$\cA$ is (respectively, convex) order closed if and only if $\cD$ is (respectively, convex) order closed (cf.~\cite[Corollary 3 and Proposition 5]{GM:18}).

\begin{lemma}\label{accep}Let $\cX$ be a function space over a fixed probability space.
\begin{enumerate}
\item\label{accep1} Let $\cD_1$ and $\cD_2$ be convex, order closed sets of $\cX_+$ that are solid in $\cX_+$. Then $\cD_1+\cD_2$ is convex and order closed in $\cX$ and is solid in $\cX_+$.
\item\label{accep2}
Let $\cA_1$ and $\cA_2$ be convex, order closed, surplus-invariant, and monotone sets in $\cX$. Then $\cA_1+\cA_2$ is convex, order closed, surplus-variant, and monotone in $\cX$.
\end{enumerate}
\end{lemma}

\begin{proof}
Clearly, $\cD_1+\cD_2$ is convex. It is also easy to check that $\cD_1+\cD_2$ is solid in $\cX_+$ by the Riesz decomposition property (\cite[Theorem~1.13]{AB:06}).
Suppose that  $(X_n)\subset \cD_1+\cD_2$ and $X\in \cX$ satisfy $ X_n\stackrel{o}{\longrightarrow}X$ in $\cX$. We want to show that $X\in\cD_1+\cD_2 $.
Write $X_n=Y_n+Z_n$, where $Y_n\in\cD_1$ and $Z_n\in\cD_2$.
Take $X_0\in \cX_+$ such that $0\leq {X_n}\leq X_0$ for all $n\in\N$. Then $0\leq Y_n\leq X_0$ and $0\leq Z_n\leq X_0$ for all $n$.
Applying Proposition~\ref{Komlos}\eqref{Komlos1} twice, we find strictly increasing $(n_j)$ and two random variables $Y,Z\in L^0$ such that $\frac{1}{k}\sum_{j=1}^kY_{n_j}\xrightarrow{a.s.} Y$ and $\frac{1}{k}\sum_{j=1}^kZ_{n_j}\xrightarrow{a.s.} Z$.
Clearly, $Y+Z=X$ and $0\leq Y,Z\leq X_0$, implying that $Y,Z\in \cX$. Since $0\leq \frac{1}{k}\sum_{j=1}^kY_{n_j}\leq X_0$ for all $k\in\N$, we have
$\frac{1}{k}\sum_{j=1}^kY_{n_j}\stackrel{o}{\longrightarrow}Y$, and thus by convexity and order closedness of $\cD_1$, $Y\in \cD_1$. Similarly, $Z\in \cD_2$.
Thus $X\in\cD_1+\cD_2$. This proves that $\cD_1+\cD_2$ is order closed.

For \eqref{accep2}, write $\cA_i=\cX_+-\cD_i$, $i=1,2$, as described preceding the lemma. Then $\cA_1+\cA_2=\cX_+-\cD_1+\cX_+-\cD_2=\cX_+-(\cD_1+\cD_2)$. By \eqref{accep1}, one sees that $\cA_1+\cA_2$ has the desired properties.
\end{proof}

Let $0<S\in\cX$. It is known (and easy to check) that $\rho$ is $S$-additive if and only if $\{\rho\leq m\}=\{\rho\leq 0\}-mS$ for every $m\in \R$ and that if $\rho$ is $S$-additive and monotone, then $\rho$ is surplus-invariant subjective to positivity if and only if $\{\rho\leq 0\}$ is surplus-invariant (\cite[Proposition~28]{GM:18}).

\begin{proposition}\label{si-thm}Let $\cX$ be a function space over a probability space, $0<S\in\cX$, and $\rho_i:\cX\rightarrow(-\infty,\infty]$, $i=1,\dots,d$, be convex, monotone, $S$-additive functionals that are surplus-invariant subject to positivity and have the (super) Fatou property . If $\Box_{i=1}^d\rho_i(X)> -\infty$ for each $X \in \cX$, then $\Box_{i=1}^d\rho_i$ is convex, monotone, $S$-additive, exact, and surplus-invariant subject to positivity and  has the (super) Fatou property.
\end{proposition}

\begin{proof}
Without loss of generality, assume $d=2$.
As remarked at the beginning of this section, $\rho_1\Box\rho_2$ is convex, monotone and $S$-additive.
Since $\{\rho_i\leq 0\}$, $i=1,2$, is convex, order closed, monotone, and surplus-invariant by the preceding remark, it follows from Lemma~\ref{accep} that $\{\rho_1\leq 0\}+\{\rho_2\leq 0\}$ is  also order closed and surplus-invariant.
%so that $\{\rho_1\Box\rho_2\leq m\}=\{\rho_1\Box\rho_2\leq 0\}-mS$ for all $m\in\R$.
We claim that $$\{\rho_1\leq 0\}+\{\rho_2\leq 0\}= \{\rho_1\Box\rho_2\leq 0\}.$$ The inclusion ``$\subset$'' is clear. For the reverse inclusion, take any $X\in\cX$ such that $\rho_1\Box\rho_2(X)\leq 0$. If $\rho_1\Box\rho_2(X)<0$, then there exist $Y,Z\in\cX$ such that $X=Y+Z$ and $\rho_1(Y)+\rho_2(Z)<0$. Take $\ep>0$ and set $Y'=Y+(\rho_1(Y)+\varepsilon)S$ and $Z'=Z-(\rho_1(Y)+\varepsilon)S$. Then $\rho_1(Y')=-\ep<0$ and $\rho_2(Z')=\rho_1(Y)+\rho_2(Z)+\ep$.
We may take $\ep$ small enough so that $\rho_2(Z')<0$ as well. Then $X=Y'+Z'\in \{\rho_1\leq 0\}+\{\rho_2\leq 0\}$. If  $\rho_1\Box\rho_2(X)=0$, then $\rho_1\Box\rho_2\big(X+\frac{1}{n}S\big)=-\frac{1}{n}<0$, so that $X+\frac{1}{n}S\in \{\rho_1\leq 0\}+\{\rho_2\leq 0\}$ for any $n\in\N$. Since $X+\frac{1}{n}S\stackrel{o}{\longrightarrow}X$, it follows that $X\in \{\rho_1\leq 0\}+\{\rho_2\leq 0\}$ by order closedness of the latter set.  This proves the claim. Consequently, $\{\rho_1\Box\rho_2\leq 0\}$ is order closed and surplus-invariant, and therefore, $\rho_1\Box\rho_2$ is surplus-invariant subject to positivity and has the Fatou, and thus super Fatou, property, by the remarks preceding the proposition and Lemma~\ref{accep}.
Finally, we prove that $\rho_1\Box\rho_2$ is exact. Pick any $X\in \cX$. If $\rho_1\Box\rho_2(X)=\infty$, there is nothing to prove. Thus assume $\rho_1\Box\rho_2(X)\in\R$. Since $\rho$ is $S$-additive, we may assume that $\rho_1\Box\rho_2(X)=0$. Then $X\in \{\rho_1\leq 0\}+\{\rho_2\leq 0\}$. Write $X=X_1+X_2$ with $X_i\in\{\rho_i\leq 0\}$, i.e., $\rho_i(X_i)\leq 0$, for $i=1,2$. It follows that $0=\rho_1\Box\rho_2(X)\leq \rho_1(X_1)+\rho_2(X_2)\leq 0$. Therefore, $\rho_1(X_1)+\rho_2(X_2)=0$.
\end{proof}

\bigskip

\begin{appendix}
\numberwithin{theorem}{section}
\section{Function Spaces}\label{R.I.}

We collect some basic notions and facts about function spaces, in particular, rearrangement invariant spaces.
%The main references are \cite[Chp.~2]{MN:91} and \cite[Chp.~1 \& 2]{BS:88}, respectively.
Fix a probability space $(\Omega,\cF,\bP)$. A {\em function space} over $(\Omega,\cF,\bP)$ is an order ideal of $L^0:=L^0(\Omega,\cF,\bP)$, i.e., a  subspace of $L^0$ such that if $X\in\cX$ and $Y$ is a random variable such that $\abs{Y}\leq \abs{X}$ then $Y\in \cX$.  A linear functional $\phi$ on a function space $\cX$ is said to be \emph{order continuous} if $\phi(X_n)\rightarrow0$ whenever $X_n\xrightarrow{o}0$ in $\cX$.
The collection of all order continuous linear functional on $\cX$ is called the \emph{order continuous dual} of $\cX$ and is denoted by $\cX_n^\sim$.
For every $\phi\in \cX_n^\sim$, there exists $Y\in L^0$ such that $\E[\abs{XY}]<\infty$ for all $X\in\cX$ and \begin{align}\label{dual}\phi(X)=\E[XY],\qquad X\in\cX;\end{align}
in fact, $Y$ is uniquely determined on the support of $\cX$.
The converse is also true, i.e., every $Y\in L^0$ such that $\E[\abs{XY}]<\infty$ for all $X\in\cX$ determines some $\phi\in\cX_n^\sim$ via \eqref{dual}.
We thus identify $\cX_n^\sim$ as a function space.
For a {\em Banach function space} $\cX$ over $(\Omega,\cF,\bP)$, i.e., a function space endowed with a complete norm such that $\norm{X}\leq \norm{Y}$ whenever $X,Y\in\cX$ and $\abs{X}\leq \abs{Y}$, it is well-known that $\cX_n^\sim$ is a Banach function space itself (cf.~\cite[Theorem 2.6.4]{MN:91}), $\cX_n^\sim\subset \cX^* $, where $\cX^*$ is the norm dual of $\cX$, and $\cX_n^\sim=\cX^*$ if and only if $\cX$ has order continuous norm. Recall that  $\cX$ has \emph{order continuous norm} if $\norm{X_n}\longrightarrow0$ whenever $X_n\stackrel{o}{\longrightarrow}0$ in $\cX$.
For a random variable $Y\in\cX_n^\sim$ we denote its norm as a linear functional on $\cX$ by
$$\norm{Y}_*=\sup\big\{\E[XY]:X\in\cX,\norm{X}\leq 1\big\}.$$

The following versions of Komlos' Theorem are very useful.
\begin{proposition}\label{Komlos}
Let $(X_n)$ be a sequence of random variables in a function space $\cX$. Then there exists a random variable $X$ (not necessarily in $\cX$) and a subsequence $(X_{n_k})$ of $(X_n)$ such that the arithmetic means of all subsequences of $(X_{n_k})$ converges to $X$ almost surely, if any of the following are satisfied:
\begin{enumerate}
\item\label{Komlos1} There exists $X_0\in L^0$ such that $\abs{X_n}\leq X_0$ for all $n\in\N$,
\item\label{Komlos2} $\cX$ is a Banach function space and $(X_n)$ is norm bounded.
\end{enumerate}
\end{proposition}

\begin{proof}
For \eqref{Komlos1}, put $\mathrm{d}\mu=\frac{1}{1+X_0}\mathrm{d}\bP$. Then $\mu$ is a finite measure on $(\Omega,\cF)$ and is equivalent to $\bP$. Since $(X_n)$ is clearly norm bounded in $L^1(\mu)$, the desired result follows from  Komlos' Theorem for $L^1(\mu) $.
For \eqref{Komlos2}, let $0\leq Y\in\cX_n^\sim$ be such that every $X\in\cX$ vanishes outside $\{Y>0\}$ up to a null set (cf.~\cite[Theorem 5.19]{GTX:17}). Put $\mathrm{d}\mu=\frac{Y+\one_{\{Y\leq 0\}}}{Y+\one}\mathrm{d}\bP$. Then $\mu$ is a finite measure on $(\Omega,\cF)$ and is equivalent to $\bP$. Moreover, $$\sup_n\norm{X_n}_{L^1(\mu)}\leq \sup_n\E[\abs{X_n}Y]\leq \norm{Y}_*\sup_n\norm{X_n}<\infty.$$
Again, the desired result follows from Komlos' Theorem for $L^1(\mu) $.
\end{proof}

\smallskip

\emph{For the rest of the appendix, we assume that $(\Omega,\cF,\bP)$ is nonatomic.} Let $\cX$ be a \emph{rearrangement invariant (r.i.)} space over $(\Omega,\cF,\bP)$, i.e., a Banach function space $\cX\neq\{0\}$ such that $X\in\cX$ whenever $X$ is a random variable that has the same distribution as some member of $\cX$. For two r.i.\ spaces $\cX$ and $\mathcal{Y}$, we write $\cX\subset \mathcal{Y}$ if every member of $\cX$ belongs to $\mathcal{Y}$, and we write $\cX=\mathcal{Y}$ or say that $\cX$ and $\mathcal{Y}$ are equal if they have the same members.
By \cite[Corollary~6.7, p.78]{BS:88}\footnote{One needs to be careful when citing \cite{BS:88} since all the Banach function spaces $\cX$ there are assumed to satisfy that $X\in \cX$ and $\norm{X}=\sup_n\norm{X_n}$ whenever ${X}$ is the a.s.-limit of a norm bounded, increasing, positive sequence with terms in $\cX$. We do not assume this extra condition, and the results we cite in this paper do not rely on this condition.}, it holds that
\begin{align}\label{contain}
L^\infty\subset \cX\subset L^1,
\end{align}
and there exist two constants $C_1,C_2>0$ such that
\begin{align}\label{iso-contain}
\norm{X}\leq C_1\norm{X}_\infty\;\;\;\forall X\in L^\infty \qquad\text{and}\qquad\norm{X}_1\leq C_2\norm{X}\;\;\;\forall X\in \cX.
\end{align}
For $t\in(0,1]$, let $$\varphi_\cX(t)=\norm{\one_E}$$ where $E\in\cF$ and $\bP(E)=t$. It is called the \emph{fundamental function} of $\cX$. Let $\cX^b$ be the norm closure of $L^\infty$ in $\cX$, and let $\cX^a$, called the \emph{heart} or the \emph{order continuous part} of $\cX$, be the collection of all $X\in \cX$ such that $\norm{X\one_{A_n}}\rightarrow 0$ whenever $A_n\downarrow \emptyset$. One can see that $X\in \cX^b$ iff $\bignorm{(\abs{X}-n\one)^+}=0$. From this it follows that $\cX^b$ is an r.i.\ space itself.

\begin{lemma}\label{aaa}
\begin{enumerate}
\item If $\lim_{t\rightarrow0^+}\varphi_\cX(t)>0$, then $\cX=L^\infty$ and $\cX^a=\{0\}$.
\item $\lim_{t\rightarrow0^+}\varphi_\cX(t)=0$ iff  $\cX^a=\cX^b$. In this case, $X^b$ is order continuous.
\end{enumerate}
\end{lemma}

\begin{proof}
(1) Suppose $\delta:=\lim_{t\rightarrow0^+}\varphi_\cX(t)>0$. Then $\norm{\one_{E}}\geq \delta$ for any $E\in\cF$ with $\bP(E)>0$. Pick any $X\in \cX$. It suffices to show that $X\in L^\infty$. If not, then $\bP(\{\abs{X}>M\})>0$ for any $M>0$. It follows from $\frac{\abs{X}}{M}\geq \one_{\{\abs{X}>M\}}$ that $\norm{X}\geq M\delta$. Letting $M\rightarrow \infty$, we get a contradiction.

(2) is \cite[Thm 5.5, p.67]{BS:88}.
\end{proof}

Since $X_n^\sim$ is also an r.i.~space (\cite[Proposition~4.2, p.59]{BS:88}), $
L^\infty\subset \cX_n^\sim\subset L^1$ as well.
Lemma~\ref{aaa}(1) applied to $\cX_n^\sim$ implies that Property $(*)$ of $\cX$ is equivalent to $\cX_n^\sim\neq L^\infty$.

\begin{proposition}\label{uo-dual}
Suppose $\cX\neq L^1$.
Suppose also that $\cX$ has order continuous norm or it contains the a.s.-limits of all norm bounded, increasing, positive sequences with terms in it. Then $\cX$ has Property $(*)$.
\end{proposition}

\begin{proof}
Suppose that $\cX$ fails Property $(*)$. Then as remarked above, $\cX_n^\sim=L^\infty$.
By Banach Isomorphism Theorem, there exists a constants $C>0$ such that
$$\norm{Y}_\infty\leq C\norm{Y}_*$$
for every $Y\in\cX_n^\sim$.
We claim that there exists $C_1>0$ such that
\begin{align}\label{l1a}\norm{X}\leq C_1\norm{X}_1\end{align}
for every $X\in \cX$.
Indeed, if $\cX$ has order continuous norm, then $\cX_n^\sim=\cX^*$. Thus for every $X\in \cX$,
$$\norm{X}=\sup_{Y\in \cX^*,\norm{Y}_*\leq 1}\E[ {XY}]\leq \sup_{Y\in L^\infty, \norm{Y}_\infty\leq C}\E[\abs{XY}]=C\norm{X}_1.$$
% In particular, it satisfies \eqref{contain} and \eqref{iso-contain} as well.
If $\cX$ contains the a.s.-limits of all norm bounded, increasing, positive sequences with terms in $\cX$, by \cite[Proposition~2.4.19(i) and Lemma~2.4.20]{MN:91}, there exists a constant $r>0$ such that
$$\norm{X}\leq r\sup_{Y\in \cX_n^\sim,\norm{Y}_*\leq 1}\E[ {XY}]\leq r\sup_{Y\in L^\infty, \norm{Y}_\infty\leq C}\E[\abs{XY}]=rC\norm{X}_1.$$
This proves the claim.
Now $\cX $ being r.i.\ also yields a constant $C_2>0$ such that \begin{align}\label{l1b}\norm{X}\geq C_2\norm{X}_1\end{align}for every $X\in \cX$.
Combining \eqref{l1a} and \eqref{l1b}, one easily checks that $\cX$ is norm closed in $L^1$. Since $\cX$ is an order ideal of $L^1$ and contains $L^\infty$, it follows that $\cX=L^1$.
\end{proof}

\begin{proposition}\label{p-star-purpose}
If $\cX$ has Property $(*)$, then $\E[X_n]\rightarrow0$ for every norm bounded sequence in $\cX$ that a.s.-converges to $0$.
\end{proposition}

\begin{proof}
Let $(X_n)\subset \cX$ be such that $M:=\sup_n\norm{X}<\infty$ and $X_n\xrightarrow{a.s.}0$.
Suppose $\E[X_n]\not\rightarrow0$. By passing to a subsequence, we may assume that $\abs{\E[X_n]}\geq \delta$ for some $\delta>0$ and all $n\in\N$. Since $X_n\xrightarrow{a.s.}0$, we can find a subsequence such that $\bP(\abs{X_{n_k}}\geq \frac{1}{k})\leq \frac{1}{k}$. Then
\begin{align*}
\abs{\E[X_{n_k}]}
\leq &\Bigabs{\E\big[X_{n_k}\one_{\{\abs{X_{n_k}}<\frac{1}{k}\}}\big]}+
\Bigabs{\E\big[X_{n_k}\one_{\{\abs{X_{n_k}}\geq \frac{1}{k}\}} \big]}\\
\leq&\frac{1}{k}\bP\Big(\abs{X_{n_k}}< \frac{1}{k}\Big)+\norm{X_{n_k}}\bignorm{
\one_{\{
\abs{X_{n_k}}\geq \frac{1}{k}
\}}
}_*\\
\leq & \frac{1}{k}+M\bignorm{
\one_{\{
\abs{X_{n_k}}\geq \frac{1}{k}
\}}
}_*\longrightarrow 0.
\end{align*}
This contradiction concludes the proof.
\end{proof}

We remark that the converse of Proposition~\ref{p-star-purpose} also holds.

\end{appendix}

\end{document}